\documentclass{stacs_proc}

\def\A{\mathcal A}
\def\B{\mathcal B}
\def\NP{\ensuremath{\mathrm{NP}}}
\def\SAT{\ensuremath{\mathrm{SAT}}{}}

\def\dNP{\ensuremath{\mathrm{DistNP}}{}}

\def\u{\underline}

\urldef{\surl}\url{http://logic.pdmi.ras.ru/{~arist,~sergey}/}

\begin{document}
\title[Complete One-Way Functions]{New Combinatorial Complete One-Way Functions}

\thanks{
Supported in part by INTAS (YSF fellowship 05-109-5565) and
RFBR (grants 05-01-00932, 06-01-00502).
}
\author[lab]{A. Kojevnikov}{Arist Kojevnikov}
\address[lab]{St.Petersburg Department of V.~A.~Steklov Institute of Mathematics
\newline Fontanka $27$, St.Petersburg, Russia, $191023$
}
\author[lab]{S. I. Nikolenko}{Sergey I. Nikolenko}
\urladdr{\surl}
\begin{abstract}
In 2003, Leonid A. Levin presented the idea of a combinatorial complete one-way function
and a sketch of the proof that Tiling represents such a function. In this paper,
we present two new one-way functions based on semi-Thue string rewriting systems
and a version of the Post Correspondence Problem and prove their completeness.
Besides, we present an alternative proof of Levin's result.
We also discuss the properties a combinatorial problem should have in order to
hold a complete one-way function.
\end{abstract}
\maketitle

\stacsheading{2008}{457-466}{Bordeaux}
\firstpageno{457}

\section{Introduction}

In computer science, complete objects play an extremely important role.
If a certain class of problems has a complete representative, one can
shift the analysis from the whole class (where usually nothing can really be proven) to 
this certain, well-specified complete problem. Examples include Satisfiability and 
Graph Coloring for \NP{} (see \cite{GJ79} for a survey) or, which is more closely 
related to our present work, Post Correspondence and Matrix Transformation
problems for \dNP{} \cite{Gur91,BG95}.

However, there are problems that are undoubtedly complete for their 
complexity classes but do not
actually cause such a nice concept shift because they are too hard to analyze.
Such problems usually come from diagonalization procedures and require enumeration of
all Turing machines or all problems of a certain complexity class.

Our results lie in the field of cryptography. For a long time, little has been known about
complete problems in cryptography. While ``conventional'' complexity classes got their
complete representatives relatively soon, it had taken thirty years since the definition
of a public-key cryptosystem \cite{DH76} to present a complete problem for the class 
of all public-key cryptosystems \cite{Harnik, GHP06}. However, this complete problem is of 
the ``bad'' kind of complete problems, requires enumerating all Turing machines and can hardly be put to any 
use, be it practical implementation or theoretical complexity analysis.

Before tackling public-key cryptosystems, it is natural to ask about a seemingly
simpler object: one-way functions (public-key cryptography is equivalent 
to the existence of a trapdoor function, a particular case of a one-way function).
The first big step towards useful complete one-way functions was taken by Leonid A. Levin
who provided a construction of the first known complete one-way function \cite{Lev87}
(see also \cite{Goldreich}).

The construction uses a universal Turing machine $U$ to compute the following function:
$$f_{uni}(\mathrm{desc}(M), x) = (\mathrm{desc}(M), M(x)),$$
where $\mathrm{desc}(M)$ is the description of a Turing machine $M$. 
If there are one-way functions
among $M$'s (and it is easy to show that if there are any, there are 
one-way functions that run in, say, quadratic time), then $f_{uni}$ is a (weak) one-way function.

As the reader has probably already noticed, this complete one-way function 
is of the ``useless''
kind we've been talking about. Naturally, Levin asked whether it is possible to find
``combinatorial'' complete one-way functions, functions that would not depend on 
enumerating Turing machines or giving their descriptions as input. For $15$ years, the problem 
remained open and then was resolved by Levin himself \cite{Lev03}. Levin devised a clever trick of 
having determinism in one direction and indeterminism in the other.

Having showed that a modified Tiling problem is in fact a complete one-way function, 
Levin asked
to find other combinatorial complete one-way functions. In this work, we answer this 
open question.
We take Levin's considerations further to show how a complete one-way function may
be derived from string-rewriting problems shown to be average-case complete in 
\cite{Wan95b} and a variation of the Post Correspondence Problem.
Moreover, we discuss the general properties a combinatorial problem should enjoy 
in order to contain a
complete one-way function by similar arguments.

\section{Distributional Accessibility problem for semi-Thue systems}

Consider a finite alphabet $\A$. An ordered pair
of strings $\langle g,h\rangle$ over $\A$ is called a \emph{rewriting rule} 
(sometimes also called a \emph{production}). We write these pairs as $g\to h$ because we interpret them as rewriting 
rules for other strings. Namely, for two strings $u$, $v$ we write $u\Rightarrow_{g\to h}v$ if $u=agb$, 
$v=ahb$ for some $a, b \in\A^*$. A set of rewriting rules is called a 
\emph{semi-Thue system}. For a semi-Thue system $R$, we write $u\Rightarrow_Rv$ if
$u\Rightarrow_{g\to h}v$ for some rewriting rule $\langle g, h\rangle\in R$. Slightly abusing notation,
we extend it and write $u\Rightarrow_Rv$ if there exists a finite sequence of
rewriting rules $\langle g_1, h_1\rangle ,\ldots,\langle g_m, h_m\rangle \in R$ such that
$$u=u_0\Rightarrow_{g_1\to h_1} u_1\Rightarrow_{g_2\to h_2}u_2\Rightarrow \ldots\Rightarrow_{g_m\to h_m}u_m=v.$$
For a more detailed discussion of semi-Thue systems we refer the reader to \cite{BO93}.

We can now define the distributional accessibility problem for semi-Thue systems:

\begin{center}\parbox{120mm}
{
\vspace{3mm}
\emph{Instance}. A semi-Thue system 
$R=\{\langle g_1,h_1\rangle,\ldots,\langle g_m,h_m\rangle\}$, 
two binary strings $u$ and $v$,
a positive integer $n$. The size of the instance is $n+|u|+|v|+\sum_1^m(|g_i|+|h_i|)$.

\emph{Question}. Is $u\Rightarrow^n_R v$?

\emph{Distribution}. Randomly and independently choose positive integers 
$n$ and $m$ and binary strings $u$ and $v$. Then randomly and independently choose binary strings 
$g_1,h_1,\ldots,g_m,h_m$. Integers and strings are chosen with the default uniform
probability distribution, namely the distribution proportional to $\frac1{n^2}$
for integers and proportional to $\frac{2^{-|u|}}{|u|^2}$ for binary strings.
\vspace{3mm}}\end{center}

In \cite{WB95}, this problem was shown to be complete for \dNP.

For what follows, we also need another notion of derivation in semi-Thue systems. 
Namely, for
a semi-Thue system $R$ we write $u\Rightarrow_R^* v$ if $u=agb$, $v=ahb$ for some 
$\langle g,h\rangle\in R$ and, moreover, there does not exist another rewriting rule $\langle g',h'\rangle\in R$ 
such that $u=a'g'b'$ and $v=a'h'b'$ for some $a',b'\in\A^*$. Similarly to $\Rightarrow_R$, we extend $\Rightarrow_R^*$ to 
finite chains of derivations. In other words, $u\Rightarrow_R^* v$ if $u\Rightarrow_R v$, and on each 
step of this derivation there was only one applicable rewriting rule. This uniqueness
(or, better to say, determinism) is crucial to perform Levin's trick. We also write
$u\Rightarrow_R^{*,n} v$ if $u\Rightarrow_R^* v$ in at most $n$ steps.

\section{Post Correspondence Problem}
The following problem was proven to be complete for \dNP{} in \cite{Gur91} 
(see also Remark $2$ in \cite{BG95}):
\begin{center}\parbox{120mm}
{
\vspace{3mm}
\emph{Instance}. A positive integer $m$, pairs
$\Gamma=\{\langle u_1,v_1\rangle,\ldots,\langle u_m,v_m\rangle\}$, a
binary string $x$,
a positive integer $n$. The size of the instance is $n+|x|+\sum_1^m(|u_i|+|v_i|)$.

\emph{Question}. Is $u_{i_1}\cdots u_{i_k} = u v_{i_1}\cdots v_{i_k}$ for some $k\leq n$?

\emph{Distribution}. Randomly and independently choose positive integers 
$n$ and $m$ and binary string $x$. Then randomly and independently choose binary strings 
$u_1,v_1,\ldots,u_m,v_m$. Integers and strings are chosen with the default uniform
probability distribution.
\vspace{3mm}}\end{center}
We need a modification of this problem. Namely, we pose the question as follows: does
$$
u_{i_1}\cdots u_{i_k} y = x v_{i_1}\cdots v_{i_k}
$$ 
hold for some $y$? If we remove the restriction $n$, this problem is
undecidable, but the bounded version is not
known to be complete for \dNP.

Given a nonempty list
$\Gamma=\left(\langle u_1,v_1\rangle, \ldots, \langle u_m,v_m\rangle\right)$
of pairs of strings, it
will be convenient to view the function based on modified Post Correspondence Problem 
as a derivation with
pairs from $\Gamma$ as inference rules. A string $x$
\emph{yields} a string $y$ \emph{in one step} if there is a pair
$\langle u,v\rangle$ in $\Gamma$ such that $uy=xv$. The
``\emph{yield}'' relation $\vdash_\Gamma$ is defined as the transitive closure of the
``yield-in-one-step'' relation.

To perform Levin's trick, we need to get rid of the indeterminism. This time, the description
of a deterministic version of $\vdash^*$ is more complicated than in the case of semi-Thue
systems. If we simply required it to be deterministic, we would not be able to move the head of
the Turing machine to the left. To solve this problem, we have to look ahead by one step:
if one of the two branches fails in two steps, we consider the choice deterministic.

Formally speaking, we write $x\vdash^*y$ if there are no more than two pairs
$\langle p,s\rangle, \langle p',s'\rangle\in\Gamma$ such that
$py=xs$ and $p'y'=xs'$ for some strings $y$, $y'$ (where $y\neq y'$, but $p$ may equal $p'$: two
possible different applications of the same rule are still nondeterministic) and,
moreover, we cannot apply any rule in $\Gamma$ to $y'$.
We write $u\vdash_\Gamma^{*,n} v$ if $u\vdash_\Gamma^* v$
in not more than $n$ steps.

\section{Complete One-Way Tiling Function}
Before presenting our own construction, we recall Levin's complete 
one-way function from \cite{Lev03}. In fact, we slightly modify
Levin's construction and present an alternative proof based on ideas from \cite{Wan99}. 
The difference with the original Levin's construction is that he 
considered the tiling function for tiles with marked corners, namely, the 
corners of tiles, instead of edges, are marked with symbols. In the tiling 
of an $n\times n$ square, symbols on touching corners of adjacent tiles 
should match. 

A \emph{tile} is a square with a symbol for a finite alphabet $\A$ on each size
which may not be turned over or rotated. We assume that there exist infinite copies of 
each tile. By a \emph{tiling} of an $n\times n$ square we mean a set of 
$n^2$ tiles covering the square in which the symbols on the common sides of 
adjacent tiles are the same.

It will be convenient for us to consider Tiling as a string transformation system.
Fix a finite set of tiles $T$. We say that $T$ \emph{transforms} a string $x$ to $y$, 
$|x|=|y|$, if there is a tiling of an $|x|\times |x|$ square with $x$ on the bottom
and $y$ on top. We write $x\longrightarrow_Ty$ in this case. By a \emph{tiling process} we mean the
completion of a partially tiled square by one tile at the time. 
Similarly to semi-Thue systems, we 
define $x\longrightarrow^*_Ty$ if and only if $x\longrightarrow_Ty$ with
an additional restriction: we permit the 
extension of a partially tiled square only if the possible extension is unique. 

\begin{definition}
The \emph{Tiling simulating function} (Tiling) is the function $f:\A^*\to\A^*$ defined as follows:
\begin{itemize}
\item if the input has the form $(T, x)$ for a finite set of tiles $T$ and a string $x$, then:
  \begin{itemize}
  \item if $x\longrightarrow^*_Ty$, then $f(T,x)=(T,y)$;
  \item otherwise, $f$ returns its input;
  \end{itemize}
\item otherwise, $f$ returns its input.
\end{itemize}
\end{definition}
\begin{theorem}
If one-way functions exist, then Tiling is a weakly one-way function.
\end{theorem}
\begin{proof}
	  
Let $Q$ be the set of states of a Turing machine $M$, $s$ be the initial state 
of $M$, $h$~--- the halting state, $\pi_M$~--- the transition function of 
$M$, $\{0, 1, B\}$~--- the tape symbols. By $\$$ we denote the begin marker and by
$\#$~--- the end marker. We also introduce a new symbol 
for each pair from $Q\times \{0,1,B\}$. We now present the construction
of a tileset $T_M$.
\begin{enumerate}
\item For each tape symbol $a\in \{0, 1, B\}$ we add
$$
\begin{picture}(150,40)
\put(0,10){\line(1,0){20}}
\put(0,10){\line(0,1){20}}
\put(20,30){\line(-1,0){20}}
\put(20,30){\line(0,-1){20}}
\put(7,0){$a$}
\put(7,35){$a$}
\put(100,10){\line(1,0){20}}
\put(100,10){\line(0,1){20}}
\put(120,30){\line(-1,0){20}}
\put(120,30){\line(0,-1){20}}
\put(99,0){$(h,a)$}
\put(99,35){$(h,a)$}
\end{picture}
$$
\item For each $a, b, c\in \{0, 1, B\}$, $q\in Q\setminus \{h\}$,
$p\in Q$, if $\pi_M(q,a)=(p,b,R)$ we add
$$
\begin{picture}(150,40)
\put(0,10){\line(1,0){20}}
\put(0,10){\line(0,1){20}}
\put(20,30){\line(-1,0){20}}
\put(20,30){\line(0,-1){20}}
\put(-1,0){$(q,a)$}
\put(7,35){$b$}
\put(23,17){$p$}
\put(100,10){\line(1,0){20}}
\put(100,10){\line(0,1){20}}
\put(120,30){\line(-1,0){20}}
\put(120,30){\line(0,-1){20}}
\put(107,0){$c$}
\put(99,35){$(p,c)$}
\put(90,17){$p$}
\end{picture}
$$

\item For each $a, b, c\in \{0, 1, B\}$, $q\in Q\setminus \{h\}$,
$p\in Q$, if $\pi_M(q,a)=(p,b,L)$ we add
$$
\begin{picture}(150,40)
\put(0,10){\line(1,0){20}}
\put(0,10){\line(0,1){20}}
\put(20,30){\line(-1,0){20}}
\put(20,30){\line(0,-1){20}}
\put(-1,0){$(q,a)$}
\put(7,35){$b$}
\put(-10,17){$p$}
\put(100,10){\line(1,0){20}}
\put(100,10){\line(0,1){20}}
\put(120,30){\line(-1,0){20}}
\put(120,30){\line(0,-1){20}}
\put(107,0){$c$}
\put(99,35){$(p,c)$}
\put(125,17){$p$}
\end{picture}
$$
\item Finally, for $\$$ and $\#$ we add
$$
\begin{picture}(150,40)
\put(0,10){\line(1,0){20}}
\put(0,10){\line(0,1){20}}
\put(20,30){\line(-1,0){20}}
\put(20,30){\line(0,-1){20}}
\put(7,0){$\$ $}
\put(7,35){$\$ $}
\put(-10,17){$\$ $}
\put(100,10){\line(1,0){20}}
\put(100,10){\line(0,1){20}}
\put(120,30){\line(-1,0){20}}
\put(120,30){\line(0,-1){20}}
\put(105,0){$\#$}
\put(105,35){$\#$}
\put(123,17){$\#$}
\end{picture}
$$
\end{enumerate}

The following lemma is now obvious.

\begin{lemma}\label{tiling-turing}
For a deterministic Turing machine $M$ that works $n^2$ steps and its 
corresponding tiling system $T_M$,
$$
M(x)=y\mbox{, $|x|=|y|$, if and only if 
$\$sxB^{n(n-1)}\#\longrightarrow^*_{T_M}\$hyB^{n(n-1)}\#$}.
$$
\end{lemma}
The rest of the proof closely follows \cite{Goldreich}. Suppose that $g$
is a length-preserving one-way function that, on inputs of length $n$, works for time not exceeding $n^2$.
By Lemma~\ref{tiling-turing}, there exists a finite system of tiles $T_M$ such that 
$\$sxB^{n(n-1)}\#\longrightarrow^*_{T_M}\$hyB^{n(n-1)}\#$ is equivalent to $g(x)=y$.
Therefore, with constant probability solving Tiling is equivalent to inverting $g$.
\end{proof}

\section{A complete one-way function based on semi-Thue systems}\label{semi-thue-owf}

Our complete one-way function is based upon the distributional accessibility 
problem for semi-Thue systems. First,
we need to make this decision problem a function and then add Levin's trick in order to assure 
length-preservation.

\begin{definition}
The \emph{semi-Thue accessibility function} (STAF) is the function $f:\A^*\to\A^*$ that
defined as follows:
\begin{itemize}
\item if the input has the form 
$(\langle g_1, h_1\rangle,\ldots, \langle g_m, h_m\rangle, x)$,
consider the semi-Thue system
$\Gamma=\left(\langle g_1, h_1\rangle,\ldots, \langle g_m, h_m\rangle\right)$ and:
  \begin{itemize}
  \item if $x\Rightarrow^{*,t}_\Gamma y$, $t=|x|^2+4|x|+2$, there are no rewriting rules in 
  $\Gamma$ that may be applied to $y$, and $|y|=|x|$, $f(\Gamma, x)=(\Gamma, y)$;
  \item otherwise, $f$ returns its input;
  \end{itemize}
\item otherwise, $f$ returns its input.
\end{itemize}
\end{definition}

Obviously, STAF is easy to compute: one simply needs to use the first part of the input as
a semi-Thue system (if that's impossible, return input) and apply its rules until either
there are two rules that apply, or we have worked for $|x|^2+4|x|+2$ steps, or $y$ has been 
reached and no other rules can be applied.
In the first two cases, return input. In the third case, check that $|y|=|x|$ and
return $(\Gamma, y)$ if so and input otherwise.

\begin{theorem}
If one-way functions exist, then STAF is a weakly one-way function.
\end{theorem}
\begin{proof}

This time we need to encode Turing machines into the string-rewriting setting. Following 
\cite{Gur91,WB95,Wan99}, we have the following proposition:

\begin{proposition}
For any finite alphabet $\A$ with $|\A|>2$ and any pair of binary strings $x$ and $y$
there exists a \emph{dynamic binary coding scheme} of $\A$ with $\{0,1\}$ with the 
following properties.
\begin{enumerate}
\item All codes (binary codes of symbols of $\A$) have the same length $l=2\log|x|+O(1)$.
\item Both strings $x$ and $y$ are distinguishable from every code, that is, no code is a 
substring of $x$ or $y$.
\item If a nonempty suffix $z$ of a code $u$ is a prefix of a code $v$ then $z=u=v$ 
(one can always distinguish where
a code ends and another code begins).
\item Strings $x$ and $y$ can be written as a unique concatenation of binary strings 
$1$, $10$, $000$, and $100$ which are not prefixes of any code.
\end{enumerate}
\end{proposition}

Now let us define the semi-Thue system $R_M$ that corresponds to a Turing machine $M$.
The rewriting rules are divided into three parts: $R_M=R_1\cup R_2\cup R_3$. Let us denote
$\B=\{1,10,100,000\}$ and fix a dynamic binary coding scheme and denote by $\u w$
the encoding of $w$ in this scheme.

$R_1$ consists of the following rules for each $u\in\B$:
$$ \begin{array}{rcl}
\u su & \to & \u{\$us_1},\\
\u{s_1}u & \to & \u{us_1},\\
\u{us_1\$} & \to & \u{s_2u\$},\\
\u{us_2} & \to & \u{s_2u},\\
\u{\$s_2} & \to & \u{\$s}.
\end{array}
$$
These rules are needed to rewrite the initial string $\u{s}x\u{\$}$ into $\u{\$sx\$}$.
Since $x$ can be uniquely written as $u_1\ldots u_m$ for some $u_i \in \B$, this 
transformation can be carried out in $2m+1\leq 2|x|+1$ steps.

$R_2$ consists of rewriting rules corresponding to 
Turing machine instructions. 
By $h$ we denote the halting state, by $s$~--- the initial state, by $B$~--- the 
blank symbol, by $Q_M$~---
the set of states of $M$, by $\pi_M$~--- the transition function of $M$, and by 
$\$$ the begin/end marker. Then $R_2$ consists of the following pairs:
\begin{enumerate}
\item For each state $q\in Q_M\setminus\{h\}$, $p\in Q$, $a,b,c\in \{0,1,B\}$:
$$\pi_M(q,a) = (p,b,R)\quad\Rightarrow\quad  qac\to bpc, 
qa\$\to bpB\$ \in R_2.$$
\item For each state $q\in Q_M\setminus\{h\}$, $p\in Q$, $a,b,d \in \{0,1,B\}$ and 
$c\in \{0,1,\$\}$,
$$\pi_M(q,a) = (p,b,L)\quad\Rightarrow\quad dqac\to pdbc,
dqB\$\to pdbB\$ \in R_2$$
for $a\neq B$, $c\neq \$$, or $b\neq B$.
\end{enumerate}

$R_1$ and $R_2$ are completely similar to the construction presented in \cite{Wan99}. 
The third
part of his construction is supposed to reduce the result from $\u{\$sy\$}$, where $y$ is the result of the
Turing machine computation, to the protocol of the Turing machine that is needed to prove that non-deterministic
semi-Thue systems are \dNP-hard.

This time we have to deviate from \cite{Wan99}: we
need a different set of rules because we actually need the output of the machine, and not
the protocol. Thus, our version of $R_3$ looks like the following:
$$ \begin{array}{rcl}
\u{\$hu} & \to & \u{\$}u\u{s_5},\\
\u{s_5u} & \to & u\u{s_5},\\
\u{s_5u\$} & \to & u\u{s_6\$},\\
u\u{s_6} & \to & \u{s_6}u,\\
\u{\$s_6} & \to & \u h.
\end{array}
$$
This transformation can be carried out in at most $2|y|+1$ steps.

These rules simply translate $\u{y}$ back into the original $y$ and add $h$ in 
front of the output,
thus achieving the actual output configuration of the original Turing machine $M$.

The following lemma is now obvious.

\begin{lemma}\label{thue-turing}
For a deterministic Turing machine $M$ and its corresponding semi-Thue system $R_M$,
$$M(x)=y\mbox{ if and only if }\u s x\u\$\Rightarrow_{R_M}^{*,t} \u hy\u\$,$$ 
where $t= T+2|x|+2|y|+2$, $T$ being the running time of $M$ on $x$. 
\end{lemma}

Again, the rest of the proof follows the lines of \cite{Goldreich}.
There is a constant probability (for the uniform distribution, it is proportional to $\frac{1}{|R|^22^{|R|}}$)
that any given semi-Thue system appears as the first part of the input. Suppose that $g$
is a length-preserving one-way function. By \cite{Goldreich}, we can safely assume that
there is a Turing machine $M_g$ that computes $g$ and runs in quadratic time. By 
Lemma~\ref{thue-turing},
there exists a semi-Thue system $R_M$ such that 
$\u s x\u\$\Rightarrow_{R_M}^{*,t} \u hy\u\$$ is equivalent
to $g(x)=y$.
Therefore, with constant probability solving STAF is equivalent to inverting $g$.
\end{proof}

\section{A complete one-way function based on Post Correspondence}\label{pcp-owf}

In this section, we describe a one-way function based on the Post Correspondence Problem and 
prove that it is complete. The function is defined as follows.

\begin{definition}
The \emph{Post Transformation function} (PTF) is the function $f:\A^*\to\A^*$ defined as follows:
\begin{itemize}
\item if the input has the form 
$(\langle g_1, h_1\rangle,\ldots, \langle g_m, h_m\rangle, x)$,
considers the derivation system
$\Gamma=\left(\langle g_1, h_1\rangle,\ldots, \langle g_m, h_m\rangle\right)$ and:
  \begin{itemize}
    \item if $x\vdash^{*,n^4}_\Gamma y$, there are no rewriting 
    rules in $\Gamma$ that may be applied to $y$, and $|y|=|x|$, then $f(\Gamma, x)=(\Gamma, y)$;
    \item otherwise, $f$ returns its input;
  \end{itemize}
\item otherwise, $f$ returns its input.
\end{itemize}
\end{definition}

Now, we reduce the computation of a universal Turing machine
to Post Correspondence in the way described in \cite{Gur91}.
\begin{theorem}
If one-way functions exist, then PTF is a weakly one-way function.
\end{theorem}
\begin{proof}
As usual, let $Q$ be the set of states of a Turing machine $M$, $s$ be the initial state
of $M$, $h$~--- the halting state, $\pi_M$~--- the transition function of
$M$, $0, 1, B$~--- the tape symbols. For all symbols we use the 
dynamic binary coding scheme described in Section~\ref{semi-thue-owf}.

We now present the construction of a derivation set $\Gamma_M$.
\begin{enumerate}
\item\label{rewriting_rule} For every tape symbol $x$: 
$$
\langle \u{x}, \u{x}\rangle.
$$
\item\label{first_rule} 
For each state $q\in Q_M\setminus\{h\}$, $p\in Q$, $a,b\in \{0,1\}$ and 
rule $\pi_M(q,a) = (p,b,R)$: 
$$
\langle \u{qa}, \u{bp}\rangle.
$$
\item For each state $q\in Q_M\setminus\{h\}$, $p\in Q$, $a\in \{0,1\}$ and     
rule $\pi_M(q,B) = (p,a,R)$: 
$$
\langle \u{qB}, \u{bpB}\rangle.
$$
\item For each state $q\in Q_M\setminus\{h\}$, $p\in Q$, $a,b,c\in \{0,1\}$ and 
rule $\pi_M(q,a) = (p,b,L)$:
$$
\langle \u{cqa}, \u{pcb}\rangle.
$$
\item\label{last_rule} 
For each state $q\in Q_M\setminus\{h\}$, $p\in Q$, $a\in \{0,1\}$ and     
rule $\pi_M(q,B) = (p,a,L)$:
$$
\langle \u{cqB}, \u{pcbB}\rangle.
$$
\end{enumerate}
The configuration of $M$ after $t$ steps of computation 
is represented by a string $xqy$, where $q$ is the current state of $M$, $x$ is the 
tape before the head, and $y$ is the tape from the head to the first blank symbol. The 
simulation of a step of $M$ from a configuration $xqy$ consists of at most $|x|$ applications of 
the rule~\ref{rewriting_rule}, followed by one application of one of the 
rules \ref{first_rule}--\ref{last_rule}, followed by $|y|-1$ applications of
rule~\ref{rewriting_rule}. Note that before an application of a rule that 
moves head to the left one could also apply rule~\ref{rewriting_rule}. If the Turing Machine
$M$ is deterministic, then this ``wrong'' application leads to a situation where no 
rule from $\Gamma_M$ is applicable. Thus, we have the following lemma.
\begin{lemma}\label{pcp-turing}
For a deterministic Turing machine $M$ with running time at most $n^2$ 
and its corresponding Post Transformation
system $\Gamma_M$,
$$M(x)=y\mbox{ if and only if }\u{sxB}\vdash_{\Gamma_M}^{*,n^4} \u{hyB}.$$
\end{lemma}

As usual, the rest of the proof closely follows \cite{Goldreich}. Suppose that $g$
is a length-preserving one-way function that works for time not exceeding $n^2$.
By Lemma~\ref{pcp-turing}, there exists a finite system of pair $\Gamma_M$ such that
$\u{sxB}\vdash_{R_M}^{*,n^4} \u{hyB}$ is equivalent to $g(x)=y$.
Therefore, with constant probability solving PTF is equivalent to inverting $g$.
\end{proof}

\begin{remark}
Note the slight change in distributions on inputs and outputs: PTF accepts as input $\u x$
and outputs $\u y$, while the emulated machine $g$ accepts $x$ and outputs $y$. 
Such ``tiny details'' often hold the devil of average-case reasoning. Fortunately,
distributions on $x$ and $\u x$ can be transformed from one to another by a polynomial algorithm, so PTF
is still a weak one-way function (see \cite{Goldreich} for details).
\end{remark}

\section{Complete one-way functions and \dNP-hard combinatorial problems}

Both our constructions of a complete one-way function look very similar
to the construction on the Tiling complete one-way function. 
This naturally leads to the question:
in what other combinatorial settings can one apply the same reasoning to find a
complete one-way function?

The whole point of this proof is to keep the function both length-preserving and
easily computable. Obvious functions fall into one of two classes.
\begin{enumerate}
\item \emph{Easily computable, but not length-preserving}. For any
\dNP-hard problem, one can construct a hard-to-invert function $f$ that transfers
\emph{protocols} of this problem into its results. This function is hard to invert
on average, but it does not preserve length, and thus it is impossible to translate
a uniform distribution on \emph{outputs} of $f$ into a reasonable distribution on
its \emph{inputs}. The reader is welcome to think of a reasonably uniform distribution
on \emph{proper} tilings that would result in a reasonably uniform distribution on their
upper rows; we believe that to construct such a distribution is either impossible
or requires a major new insight.

\item \emph{Length-preserving, but hard to compute}. Take a \dNP-hard problem and
consider the function that sends its input into its output (e.g. the lowest
row of the tiling into its uppermost row). This function is hard to invert and
length-preserving, but it is also hard to compute, because to compute it one needs
to solve Tiling.
\end{enumerate}

Following Levin, we get around these obstacles by having a \emph{deterministic}
version of a \dNP-hard problem. This time, a Tiling problem produces nontrivial
results only if there always is \emph{only one} proper tile to attach. Similarly,
in Section~\ref{semi-thue-owf} we demanded that there is only one rewriting rule
that applicable on each step (we introduced $\Rightarrow^*$ for this very purpose).
In Section~\ref{pcp-owf} we slightly generalized this idea of determinism, allowing 
fixed length deterministic backtrack.
However, if for all $z\in f^{-1}(y)$ we can do this deterministic procedure, then
we can easily invert $f$. So we need that for most $z$ an indeterminism appears and 
the procedure return $z$.


A combinatorial problem should have two properties in order to hold a complete 
one-way function.
\begin{enumerate}
\item It should have a deterministic restricted version,
like Tiling, string rewriting and modified Post Correspondence.
\item Its deterministic version should be powerful enough to simulate a deterministic
Turing machine. For example, natural 
deterministic Post Correspondence (without any backtrack) is, of course, easy to formulate,
but does not seem to be powerful enough.
\end{enumerate}

Keeping in mind these properties, one is welcome to look for other combinatorial
settings with combinatorial complete one-way functions.

\section{Discussion and further work}

We have shown a new complete one-way function and discussed 
possibilities of other combinatorial settings to
hold complete one-way functions. These functions are combinatorial 
in nature and represent a step towards the 
easy-to-analyze complete cryptographic objects, much like 
\SAT{} is a perfect complete problem for \NP.

However, we are still not quite there. 
Basically, we sample a Turing machine at random and hope to find precisely the hard one.
This distinction is very important for practical implications of our constructions.
We believe that constructing a complete cryptographic problem that has properties
completely analogous to SAT
requires a major new insight, and such a construction represents one of the most important
challenges in modern cryptography.

Another direction would be to find other similar combinatorial problems that can 
hold a complete one-way function. By looking at our one-way functions 
and Levin's Tiling, one could imagine that every
\dNP-complete problem readily yields a complete one-way function. However, there 
is also this subtle
requirement that the problem (or its appropriate restriction) should be deterministic
(compare $\Rightarrow_R^*$ and $\Rightarrow_R$). It would be interesting to restate this
requirement as a formal restriction on the problem setting. This would require some new
definitions and, perhaps, a more general and unified approach to combinatorial problems.

\section*{Acknowledgments}
The authors are very grateful to Dima Grigoriev, Edward A. Hirsch and 
Yuri Matiyasevich for helpful comments and fruitful discussions.

\bibliographystyle{alpha}
\newcommand{\etalchar}[1]{$^{#1}$}

\end{document}